\documentclass[conference]{IEEEtran}
\IEEEoverridecommandlockouts
\usepackage{cite}
\usepackage{amsmath,amssymb,amsfonts}
\usepackage{algorithmic}
\usepackage{graphicx}
\usepackage{textcomp}
\usepackage{xcolor}
\usepackage{multirow}
\def\BibTeX{{\rm B\kern-.05em{\sc i\kern-.025em b}\kern-.08em
    T\kern-.1667em\lower.7ex\hbox{E}\kern-.125emX}}

\graphicspath{{figures/}, {results/}}
\DeclareGraphicsExtensions{.pdf,.jpeg,.png,.eps}

\usepackage{url}
\usepackage{amsmath,amsthm,amssymb, algorithmic}
\usepackage[lined, boxed, comments numbered, ruled, linesnumbered]{algorithm2e}
\newtheorem{theorem}{\textbf{Theorem}}
\newtheorem{corollary}{\textbf{Corollary}}[theorem]

\def\lthe#1{\textsc{Theorem~\ref{the:#1}}}
\def\lcor#1{\textsc{Corollary~\ref{cor:#1}}}

\def\lieq#1{\textsc{Inequation~\ref{eq:#1}}}
\def\lag#1{\textsc{Algorithm~\ref{algo:#1}}}

\def\lcom#1{$\mathcal{O}$(#1)}

\newif\ifmark
\newif\ifhidenote
\marktrue
\hidenotefalse

\newif\ifspace
\spacetrue

\usepackage[normalem]{ulem}
\ifmark

  \newcommand{\del}[1]{\sout{#1}}
  \ifhidenote
     \newcommand{\note}[1]{}
  \else
      \newcommand{\note}[1]{{\sffamily\itshape\bfseries\uline{#1}}}
  \fi
\else
  
  \newcommand{\del}[1]{}
  \newcommand{\note}[1]{}
\fi

\vspace{-12cm}  
\begin{document}

\title{TupleChain: Fast  Lookup of OpenFlow Table with Multifaceted Scalability}

\author{
\IEEEauthorblockN{
Yanbiao Li\IEEEauthorrefmark{1}\IEEEauthorrefmark{2},
Neng Ren\IEEEauthorrefmark{1}\IEEEauthorrefmark{2},
Xin Wang\IEEEauthorrefmark{3},
Yuxuan Chen\IEEEauthorrefmark{1}\IEEEauthorrefmark{2},
Xinyi Zhang\IEEEauthorrefmark{1}\IEEEauthorrefmark{2},
Lingbo Guo\IEEEauthorrefmark{1}\IEEEauthorrefmark{2},
and Gaogang Xie\IEEEauthorrefmark{1}\IEEEauthorrefmark{2}}
\IEEEauthorblockA{\IEEEauthorrefmark{1} 
\textit{Computer Network Information Center},
\textit{Chinese Academy of Sciences}, Beijing, China}
\IEEEauthorblockA{\IEEEauthorrefmark{2} 
\textit{School of Computer Science and Technology},
\textit{University of Chinese Academy of Sciences}, Beijing, China}
\IEEEauthorblockA{\IEEEauthorrefmark{3} 
\textit{Department of Electrical and Computer Engineering},
\textit{Stony Brook University}, Stony Brook, NY, USA}
}

\maketitle

\begin{abstract}
OpenFlow switches are fundamental components of software
defined networking, where the key operation is to look up flow tables
to determine which flow an incoming packet belongs to. This
needs to address the same multi-field rule-matching problem as legacy
packet classification, but faces more serious scalability
challenges. The demand of fast on-line updates makes most
existing solutions unfit, while the rest still lacks the scalability
to either large data sets or large number of fields to match for a rule. In this work, we
propose TupleChain for fast OpenFlow table lookup with
multifaceted scalability. We group rules based on their masks, each
  being maintained with a hash table, and explore the connections
  among rule groups to skip unnecessary hash probes for fast search. We show via
  theoretical analysis and extensive experiments that the proposed
  scheme not only has competitive computing complexity, but is also scalable
  and can achieve high performance in both search and update. It can process
multiple millions of packets per second, while dealing with millions
of on-line updates per second at the same time, and its lookup
speed maintains at the same level no mater it handles a large flow table
with 10 million rules or a flow table with every entry having as many
as 100 match fields.
\end{abstract}

\begin{IEEEkeywords}
OpenFlow, lookup, on-line updates, scalability
\end{IEEEkeywords}

\section{Introduction}
\label{sect:intro}

\subsection{Background}
\label{sect:background}
Software Defined Networking (SDN)~\cite{sdn} not only offers richer flexibility to operators but also adds more scalability and programmability to forwarding devices. As the de facto standard of SDN, OpenFlow~\cite{of-specification, openflow} defines how the control plane (i.e., controllers) and the forwarding plane (i.e., switches) collaborate, and how a switch processes packets.

An OpenFlow switch contains a sequence of flow tables, where every entry of which is identified by a number of match fields and a match precedence used to decide the priority in case of multiple matches. The core operation of an OpenFlow switch is to look up its flow tables to find the best flow entry that matches the incoming packet, and in case of multiple matches, the one with the highest priority is selected. Generally, lookups on flow tables are pipelined. Prior studies have dramatically reduced the cost in pipeline~\cite{ovs, BFI}, driving our focus to the current performance bottleneck --- flow table lookup (flow lookup in short).


\subsection{Problem Statement and Scalability Challenges}
\label{sect:scalability}
Given a flow table composed of rules each with $d$-field, for an incoming packet, flow lookup is performed by matching the $d$-field key extracted from the packet header (and optionally some metadata) against the flow table to find a rule that matches this packet and has the highest priority. The problem of multi-field rule-matching has been well studied for packet classification. However, in addition to the high speed, flow lookup poses stronger demands on the scalability. 

First, a flow lookup scheme must work well in the presence of highly frequent rule updates~\cite{kuzniar2015you}. Generally, to reduce their influence on the lookup process, rule updates must be performed as quickly as they arrive. However, this is exactly what most packet classification algorithms~\cite{BytesCut,  class:cut:effi, class:cut:hyper} lack. Recent attentions have been drawn to both fast lookup and fast updates~\cite{TM, MT, TupleTree, CutTSS}.

Second, a flow lookup scheme should work well with large data sets.
In a typical data center, an edge switch may have to handle more than $1$ million flows~\cite{scale-cloud}, while a gateway router at the border of autonomous system  may handle about $0.8$ million forwarding rules~\cite{scale-bgp}. Intuitively, a future-proof flow lookup scheme should work well when handling millions of rules. However, this scale is 2 orders of magnitude larger than the largest data sets tested with existing solutions~\cite{TM, MT, TupleTree, CutTSS}.

Lastly, a flow lookup scheme should work well with rules having a large number of fields. OpenFlow has defined 45 match fields in its specification 1.5~\cite{of-specification}, which will  increase in number when new protocols are supported. Though OpenFlow switches will not deal with all match fields in one table, addressing this larger number of match fields at the algorithmic level will offer more flexibility and opportunities for system optimization. However, to our best knowledge, none of existing solutions have been verified to be able to work on rules with a large number of fields (i.e., at the scale of 100 or so).

\subsection{Our Contributions}
\label{sect:solution}
In this paper, we propose \emph{TupleChain}, a novel flow lookup scheme with both fast lookup and efficient updates, in view of not only computation complexity but also practical performance. Most importantly, its excellent scalability in the aforementioned three aspects is clearly demonstrated. We summarize the main contributions of this paper as follows:

\begin{enumerate}
\item We propose the use of a directed acyclic graph to track the connections between rule groups created following the \emph{tuple space search} model~\cite{tss}. With every rule group referred as a \emph{tuple}, we name           this graph a \emph{tuple graph}. On this basis, we introduce two types of lookup guidance, where the tuple connections thus the edges of \emph{tuple graph} are exploited to skip unnecessary searches in flow lookup.
\item We propose the use of tuple chains to trade off between the skipping of search operations and the maintaining of lookup guidance information. We group edges in \emph{tuple graph} into several tuple chains, where every chain is unidirectional and follows a monotonic sequence.
\item We present a series of algorithms based on tuple chains for flow lookup, rule updates and maintenance of guidance information.
\item We analyze the complexity to show that our scheme \emph{TupleChain} supports fast lookup and fast updates with the cost effectively amortized.
\item We extend  \emph{TupleChain} to further boost its performance with the optimal construction of  tuple chains, the adjustment of inner structure of  a tuple chain during tuple insertion,  as well as the increase of runtime speed and scalability.
\item We evaluate the performance of basic \emph{TupleChain} and extended scheme. Our proposed \emph{TupleChain}  is demonstrated to be able to handle  extremely high update frequency ($1$ million updates per second),          super large data sets ($10$ million rule sets) and a large number of fields ($100$ fields). When the scale for each becomes large in the experiments, our scheme is the only survivor in all cases, while keeping the             system throughput higher than $1$ million packets per second all the time.
\end{enumerate}

The rest of this paper is organized as follows. Section~\ref{sect:tss} reviews the literature work. Section~\ref{sect:tc} presents our motivation, core ideas and the basic scheme of \emph{TupleChain}, which is followed by a comprehensive complexity analysis in Section~\ref{sect:analysis}. Section~\ref{sect:optimization} introduces a series of technics to boost \emph{TupleChain}'s practical performance. Then we evaluate \emph{TupleChain} and some state-of-arts in Section~\ref{sect:exp} Finally, Section~\ref{sect:end} concludes the whole paper.
\section{Related Work}
\label{sect:tss}
We first review two categories of work related: 1) packet classification and 2) OpenFlow table lookup, and then summarize the differences between our proposal and related solutions.

\subsection{Packet Classification}
\label{sect:work:pc}
Hardware-based classifiers are widely adopted in industry. TCAM (Ternary Content Addressable Memory)-based solutions offer very fast speed, but their slow updates~\cite{tcam3} restrict their use.  With carefully designed structures and pipelines, FPGA (Field-Programmable Gate Array)-based solutions~\cite{fpga1,zhong2021efficient} are faster and more flexible than TCAM.

Most algorithmic solutions target to the software scenario. In the early days, in addition to trie-based solutions~\cite{class:trie:grid}, Cross-Producting~\cite{class:trie:grid} and Recursive Flow Classification~\cite{class:pro:rfc} attracted lots of interestes for their fast speed. However, neither of them works well with large data sets. Current state-of-the-art solutions are based on decision tree~\cite{ class:cut:effi, BytesCut}. They achieve fast speed with heuristic strategies at the cost of slow updates, and their performance varies a lot across different data sets~\cite{he2014meta}.

\subsection{OpenFlow Table Lookup}
\label{sect:work:of}
Many classification algorithms only work with static sets of flows, or have expensive incremental update procedures, making them unsuitable for dynamic OpenFlow flow tables due to their slow updates. For a better trade-off between lookup and update, \emph{PartitionSort}~\cite{PS} divides rules into sortable rulesets, which supports both fast search and fast update by utilizing the binary search tree. On the other hand, the \emph{Bloom Filter Intersection (BFI)}~\cite{BFI} follows the basic mode of \emph{Bit Vector}~\cite{BV}, but represents the results on individual fields as bloom filters. It can achieve fast lookup with efficient updates. However, it can not scale well to the number of rules due to the fixed size of bloom filters.

\emph{Tuple space search (TSS)}~\cite{tss} is designed for packet classification but is well suited to flow lookup. Its core idea is to divide a large table into groups, where rules in the same group share the same mask for the fields to match. A flow lookup needs to search all the groups with a hash probe on each, and output the best result. This scheme is proposed with several extensions, among which the \emph{pruned tuple search (PTS)} is the fastest. It processes individual fields and combines the results to pick up the candidate groups to search in the next step. In contrast, \emph{tuple search using a balancing heuristic (TSBH)} focuses on reducing the complexity. By repeatedly selecting the best group to probe with a balancing heuristic, and skipping part of the  remaining groups according to the search on the selected one, the number of required probes can be sharply reduced. Its lookup complexity is $\mathcal{O}(m^{\log_32})$, where $m$ is the number of groups. 

Open vSwitch~\cite{ovs} adopts the basic \emph{TSS} scheme and improves its practical performance with a series of runtime pruning . We refer this scheme as \emph{Priority Sorted Tuple Search (PSTS)}. \emph{TupleMerge (TM)}~\cite{TM} aims to reduce the number of rule groups at the construction time. Its core idea is to move the rules in some groups to others to restrict the collision rate below a threshold. \emph{MultilayerTuple (MT)} ~\cite{MT} and \emph{TupleTree} ~\cite{TupleTree} inherit this merging approach. Both methods merge all tuples into several "big" tuples, which causes collisions. So they rearrange the rules at collision entry into a substructure, forming a "multilayer" or "tree" architecture. The difference between the two methods is the way of merging. \emph{MT} adopts a static approach while \emph{TupleTree} uses a heuristic one. \emph{CutTSS}~\cite{CutTSS} exploits the joint use of decision tree and TSS, where it first divides the rule set into several groups  and then uses TSS for the groups that contain many overlapping rules.

\subsection{Summary of Prior Arts and Our Solutions}
Because of its generality of fields, linear memory cost and constant update complexity, the \emph{TSS} model has been proven to be a good starting point to develop a flow lookup scheme with multifaceted scalability. However, its performance may suffer when the number of groups is large, as it has to probe all groups. \emph{PTS}, \emph{TSPS}, \emph{TM}, \emph{MT}, \emph{TupleTree} all aim at improving the practical performance, but none of them provides the worst-case performance guarantee since they have the same complexity as \emph{TSS}. Although \emph{TSBH} makes a great effort to lower the lookup complexity, its practical performance is not that good and its update is too complicated. \emph{CutTSS} gains performance via cutting but its update and memory cost also deteriorates.
In this work, we start from \emph{TSS} as well, but propose a new scheme \emph{TupleChain} to conquer the performance challenge. By exploiting the connections between rule groups and carefully trading off between the lookup speed and update speed, our approach achieves both fast lookup and fast update while guaranteeing  the worst-case lookup performance and average update performance.

\section{TupleChain: Binary Search on Chained Tuples}
\label{sect:tc}

In this section, we first introduce the basic model
and our motivation in Section~\ref{sect:tc:motivation}, then the
concepts of tuple graph and lookup guidance information in Section~\ref{sect:tc:graph}. We further present the basic scheme of \emph{TupleChain}, as well as its algorithms for packet lookup and rule update in Section~\ref{sect:tc:chains}.

\subsection{Basic Model and  Motivation}
\label{sect:tc:motivation}
We denote a $d$-field rule $r$ as ($\vec f, \vec m, pri$),
where the integer $pri$ denotes the rule's priority, the $d$-dimensional vectors $\vec f$ and $\vec m$
represent its field to match and mask respectively.  Before
a flow lookup, the search key $\vec k$ with the corresponding $d$ fields is generated based on the incoming
packet $p$ and some metadata. A packet $p$ matches a rule $r$
if and only if $\vec k~\&~\vec m = \vec f$.
 Following the basic tuple space search (\emph{TSS}) model, a flow table
is divided into \emph{tuples} (as shown in Fig.~\ref{fig:rules} and
  Fig.~\ref{fig:tuples}), each is identified by a mask (a $d$-dimensional vector) and associated with a hash table (keyed by $d$-dimensional
vectors). For simplicity, we refer to an entry of a tuple's
  hash table as the ``tuple's entry''.

With \emph{TSS}, flow lookup is performed by searching all tuples and returning the one
with the highest priority among all matched rules. Fig.~\ref{fig:lookup:tss}
shows the \emph{TSS} constructed with the rules shown in
  Fig.~\ref{fig:rules}, where every tuple is denoted as
a labelled cycle. When processing a packet, all 6 tuples are searched.

\begin{figure}
  \begin{minipage}[b]{0.6\linewidth} 
    \centering
    \includegraphics[width=\linewidth]{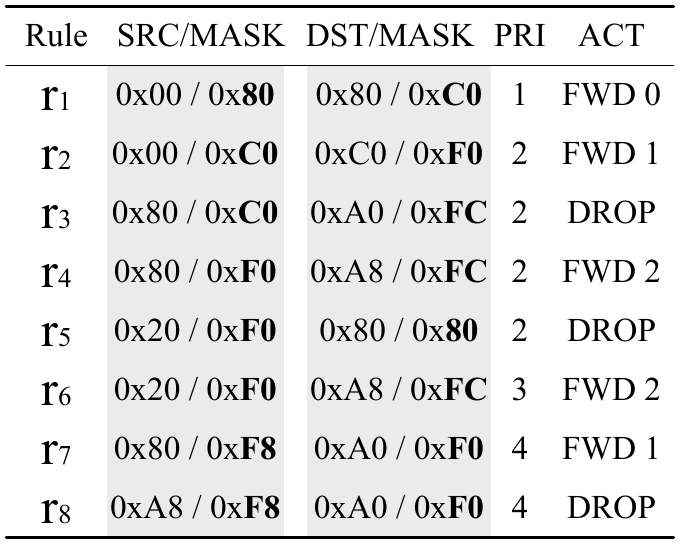}
    \caption{simplified 2-field rules}
    \label{fig:rules}
  \end{minipage}%
  \begin{minipage}[b]{0.05\linewidth}\end{minipage}%
  \begin{minipage}[b]{0.35\linewidth}
    \centering
    \includegraphics[width=\linewidth]{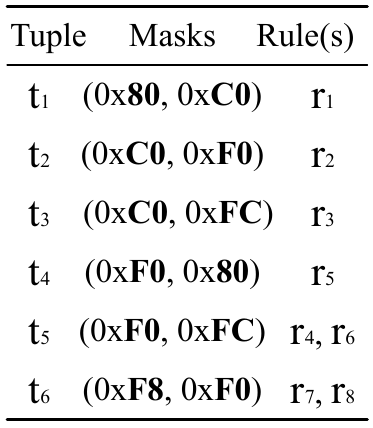}
    \caption{constructed tuples}
    \label{fig:tuples}
  \end{minipage}
\end{figure}

\begin{figure}[tbp]
  \centering
  \begin{minipage}[b]{0.47\linewidth} 
    \centering
    \includegraphics[width=\linewidth]{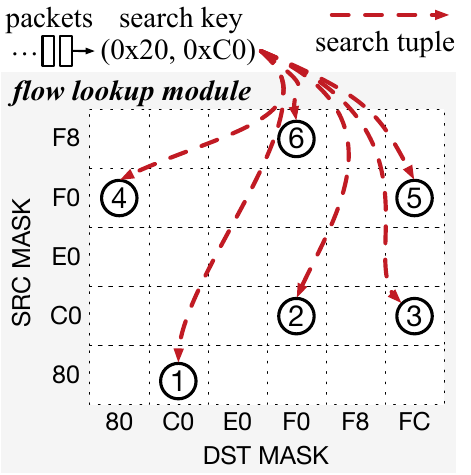}
    \caption{lookup with \emph{TSS}.}
    \label{fig:lookup:tss}
  \end{minipage}%
  \quad
  \begin{minipage}[b]{0.47\linewidth}
    \centering
    \includegraphics[width=\linewidth]{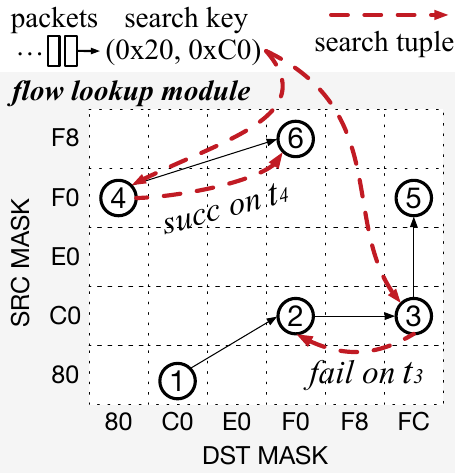}
    \caption{lookup with \emph{TupleChain}.}
    \label{fig:lookup:tc}
  \end{minipage}
\end{figure}

Our basic idea is to exploit the connections between tuples to avoid unnecessary searches.
   We propose the use of  \emph{TupleChain} to organize tuples into a
set of chains, where two consecutive tuples on a chain have
a unidirectional connection. Though all chains
will be searched, the lookup on each chain can be
well-organized to skip unnecessary searches.
In the example of Fig.~\ref{fig:lookup:tc}, 6 tuples form 2 chains to be searched for
the incoming packet. The lookup on the first chain starts from
  $t_{4}$ and ends after searching $t_{6}$. When searching along the second chain, the
  miss of the first probe on $t_{3}$ directs the lookup to
  $t_{2}$. $t_{5}$ is skipped, because all its entries have left their
  \emph{markers} on $t_{3}$ and the miss on $t_{3}$ indicates that the
  markers of $t_{5}$'s entries  also cannot be
  matched. 



As the probe on $t_{2}$ succeeds, we can skip $t_{1}$, because all the rules from $t_{1}$ that could be matched by the incoming packet must
  have been reported as \emph{hints} to $t_{2}$. Accordingly, the
  lookup on this chain is terminated. We will introduce the details of
  markers and hints in the Section~\ref{sect:tc:graph}.

For this approach to work, we need to answer a set of questions: 1) How to set up and make
use of the connections between tuples? 2) How to construct chains? 3)
How to organize the search along each chain and how fast will the
search be? How to perform rule updates without impacting the
connections between tuples? We will answer these questions in the rest of this section.

\subsection{Tuple Graph and Lookup Guidance}
\label{sect:tc:graph}

We first introduce the \emph{tuple graph}, a directed acyclic
  graph that tracks the connections between any pair of tuples, as
  well as two types of information, \emph{markers} and \emph{hints},
  to guide more efficient flow lookups with a tuple graph.

Given two tuples $t_{x}$ and $t_{y}$, if $t_{x}$'s mask is contained
in $t_{y}$'s on every field, we denote this relation as $t_{x} <
t_{y}$. In Fig.~\ref{fig:tuples}, $t_{1}<t_{2}$ because every field of $t_{1}$'s mask (0x80, 0xC0), i.e.,($\textbf{1}0000000, \textbf{11}000000$) in binary format, is contained in
the corresponding field of $t_{2}$'s mask (0xC0, 0xF0), i.e.,
($\textbf{11}000000, \textbf{1111}0000$). The prefix length associated with the mask of $t_{1}$ is (1, 2) and the prefix length of  $t_{2}$ is  (2, 4). Obviously, rules in $t_{2}$ are more specific and cannot be matched if a search cannot match ones in  $t_{1}$. Given a set of tuples,
  we construct a \emph{tuple graph} by making every tuple a vertex,
  and adding an edge from $t_{x}$ to $t_{y}$ if $t_{x} < t_{y}$. On
the tuple graph, the search of tuples is transformed into the traverse of
vertexes.

To reduce the vertex thus tuple to visit in performing the flow lookup with
  the tuple graph, we leave and apply {\em markers} and {\em hints}
  along its edges backward and forward respectively. Given an edge
  from $t_{x}$ to $t_{y}$, from any entry $e_{y}$ of $t_{y}$ we
  create an entry $e_{x}$ and insert it into $t_{x}$,
  whose key is made by masking the key of $e_{y}$ with the mask of
  $t_{x}$. This makes the key of $e_{x}$ part that of $e_{y}$. We call $e_{x}$ the \emph{marker} of $e_{y}$ and $e_{y}$
  the \emph{owner} of $e_{x}$. As an entry,  a marker can also hold a rule belonging to the tuple $t$ it is inserted into.
   It can further leave markers in other tuples that have edges to $t$, one for
  each tuple. Besides, multiple entries from
  different tuples can share the same marker in a tuple as well. In the example of
Fig.~\ref{fig:hints}, along the edge from $t_{3}$ to $t_{5}$, an entry
$e_{1}$ in $t_{5}$ can leave a marker in $t_{3}$. By masking its key
(0x20, 0xA8) with $t_{3}$'s mask (0xC0, 0xFC), a new key (0x00, 0xA8)
is generated to associate with the entry $e_{2}$ in
$t_{3}$. Initially, $e_{2}$ does not hold any rule that belongs to $t_{3}$
but can further leave its markers to $t_{2}$ ($e_{3}$) and $t_{1}$
($e_{4}$) respectively.
The key~(0x20, 0xA8) in $t_{5}$  is more specific than the key (0x00, 0xA8) in $t_{3}$,
so we can get an important property for markers: {\em If a packet
succeeds in matching an entry in a tuple, it must be able to match all markers of this entry. 
  On the contrary, if a packet fails to match an entry, it must be unable to match all owners of this entry.} 



On the other hand, along an edge from $t_{x}$ to $t_{y}$, any marker in $t_{x}$ can report
a \emph{hint} to its owner in $t_{y}$ to help cut short the search path. 
A hint of  a marker is the best rule that can be matched by a packet with this marker and all of its own markers. More specifically, the rule held in an entry (if
  any) and all hints the entry received from its markers form the candidates
  to determine the entry's \emph{best rule} for match,  from which the one with the highest priority is selected. In Fig.~\ref{fig:hints},
two markers $e_{3}$ and $e_{4}$ report their hints, $r_{2}$ and $r_{1}$ respectively, to their owner $e_{2}$. Since $e_{2}$ does not
hold any rule in $t_{3}$, it selects $r_{2}$ (assuming $r_{2}$ has a higher priority than $r_{1}$) as its \emph{best rule}, which is further reported to its owner $e_{1}$. However, $e_{1}$ holds a rule $r_{6}$
whose priority is higher than $r_{2}$, so $e_{1}$'s \emph{best rule} would stay as $r_{6}$. With the hint, we have another important
property: {\em If a packet matches an entry in a tuple, there is
no need to search and match against the tuples hosting its markers, since  the best match with the corresponding rule has already been included  by the hint from its markers.} 

\begin{figure}[tbp]
  \centering
  \begin{minipage}[b]{0.46\linewidth} 
    \centering
    \includegraphics[width=\linewidth]{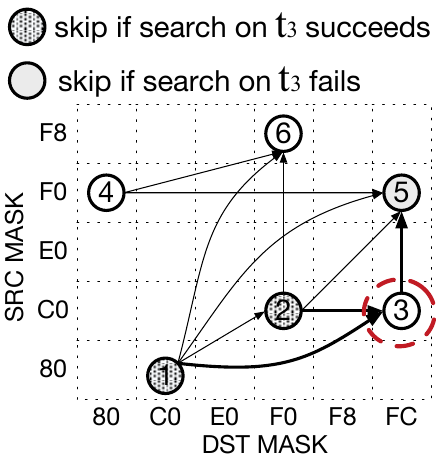}
    \caption{a tuple graph}
    \label{fig:graph}
  \end{minipage}%
  \begin{minipage}[b]{0.54\linewidth}
    \centering
    \includegraphics[width=\linewidth]{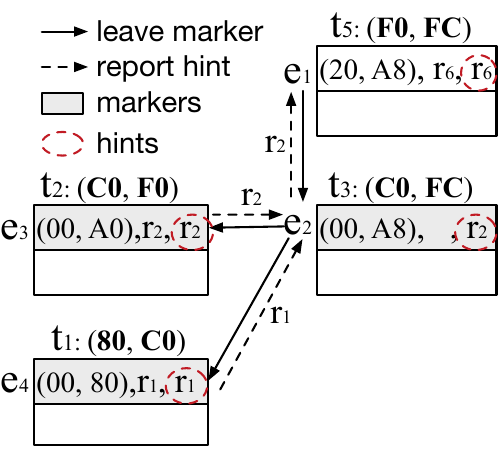}
    \caption{part of lookup hints for $t_{3}$.}
    \label{fig:hints}
  \end{minipage}
\end{figure}

Based on these two types of guiding information, we can reduce unnecessary
search. Given an edge from $t_{x}$ to $t_{y}$ on a tuple
  graph, for every entry in $t_{y}$, we leave its marker in $t_{x}$
  and update its \emph{best rule}.
    We can make the following reductions. In the case that $t_{x}$ is searched first but
    the probe fails, it's safe to skip $t_{y}$ because all its
    entries have left their markers  in $t_{x}$. A further search is not needed if a search with a coarser prefix fails.     On the other hand,  when $t_{y}$ is checked first and it gives a match with the
    entry $e$, we can skip $t_{x}$ as the only entry the packet can
    match from $t_{x}$ must be $e$'s marker, which has reported its hint to $e$. 




Fig.~\ref{fig:graph} shows an example
of utilizing the lookup hints in $t_{3}$ to avoid unnecessary
searches. Once the search on $t_{3}$ succeeds, the searches in tuples $t_{1}$ and $t_{2}$ can be
skipped,  while a mismatch of $t_{3}$ allows us to skip the lookup in $t_{5}$. 

\subsection{Maintenance of Markers and Hints}
\label{sect:tc:maintain}
To maintain markers and hints, we extend the
design of the hash table entries. In our scheme,
every entry $e$ of tuple $t$'s hash table is composed of a
$d$-field \emph{key} (as the identity), a \emph{rule} belonging
  to $t$, a \emph{hint} and an \emph{owner list}. To save space
  and for the storage alignment, the \emph{rule} and the \emph{hint} are
  the pointers that point to corresponding rules, while the \emph{owner
    list} is the pointer pointing to a linked list that stores the
  pointers pointing to all entries sharing $e$ as a marker. An empty
  owner list indicates that the entry is not a marker.

An entry $e$ of a tuple $t$ will be created in two cases: 1) inserting a rule
belonging to $t$ or 2) inserting another entry's marker into
$t$. In the first case, $e$'s rule is set as the one
  being inserted while its hint and owner list are left
  empty. The owner list will be updated when an entry of another tuple leaves its marker
    to $e$ of the tuple $t$ at a later time.
    In the second case, $e$ is created to host the marker of an entry from another tuple, with the entry inserted into its owner list. The rule and the
  hint of $e$  are left empty initially, and the rule will be updated when a rule is inserted into $e$ of $t$.


%
%

\subsection{Tuple Chains and the Lookup Algorithm}
\label{sect:tc:chains}

Benefiting from lookup guidance, flow lookup with the tuple graph
  can be performed more efficiently. However, there are two
  drawbacks. First, once the search of a tuple is done, current lookup
  guidance can tell which tuples can be skipped, but not which one is
  the best to search in the next step. Actually, for different lookup
  requests there may be different optimal probing paths on the tuple
  graph. It's hard to pre-compute such optimal paths for future
  lookups. Besides, maintaining too much lookup guiding information will
  make rule updates extremely complicated.

To address these issues, we propose the construction of \emph{TupleChain}
  where we break a tuple graph into disjoint chains that cover all
  tuples. This scheme brings in three benefits:
\begin{enumerate}
\item All tuples in a chain form a monotonic
  sequence with the ``$<$'' operator, enabling an efficient binary search.
\item Every rule update involves only a single chain, thus the update can be kept within this chain.
\item A tuple has at
  most one preceding tuple and one succeeding tuple. Thus, an
  entry has at most one marker, facilitating the marker
  maintenance. Although a marker can still be shared by multiple
  entries where it has to report hints, the overall cost across all
  tuples on the same chain can be amortized (see the proof in Section~\ref{sect:analysis}).
\end{enumerate}

Essentially, all these chains form a disjoint path cover of the
  tuple graph. We propose an  optimal method to form chains in
  Section~\ref{sect:construct-adaptation}. Every chain is maintained as a red-black tree. For each tuple node, we rename its left and right children pointers as ``fail'' and ``succ'' respectively, and add two additional
pointers ``prev'' and ``next'', to point to its preceding and
succeeding tuples to facilitate the maintenance of markers and
hints.

The flow lookup with \emph{TupleChain} is simple. Every
chain is searched by performing a tree traversal, and the output is the best
result returned by searches from all chains. As described in \lag{lookup:basic}, the
search on every chain starts from the tree root. In every step, the
search is directed to the next node following the ``succ'' pointer or stops following
the ``fail'' pointer, based on whether the current node yields a
match or not.

\begin{algorithm}[htbp]
\caption{\label{algo:lookup:basic} flow Lookup with \emph{TupleChain}}
\small
\LinesNumbered
\KwIn{$packet$}
\KwOut{$bestRule$}
\BlankLine
$bestRule =$ DEFAULT\_RULE;\\
\ForEach{$chain$}{
        $tp = chain.root$;\\
        \While{$tp$}{
                  $e = tp.table.$search ($packet.\vec k$);\\
                  CHECK\_UPDATE\_BEST\_RULE($bestRule,~e$);\\
                  $tp = e~?~tp.succ~:~tp.fail$;\\
        }
}
\end{algorithm}

\subsection{Rule Updates with TupleChain}
\label{sect:tc:update}
Here, we discuss how to update a rule with
\emph{TupleChain}. We first introduce rule insertion and rule deletion
with high level logics, and then dive into details of dealing with
markers and hints. We close this subsection
with tuple insertion / deletion.

\subsubsection{Rule insertion}
\label{sect:tc:update:insert}

When inserting a rule $r$,  we shall insert it into an entry $e$ of its corresponding tuple $t$ , and then build the marker link and update the hint of $e$. It starts  by finding out tuple t which has the same mask with $r$. If entry $e$ does not exist, it will be created and leave marker in the preceding tuple of chain. Then $e$'s rule and hint are set to $r$. The marker of $e$ returns its hint, which can be used to update $e$'s hint. This updated hint will be further reported to $e$'s owner(s).

\begin{algorithm}[htbp]
\caption{\label{algo:insert:rule} insert a rule with \emph{TupleChain}}
\small
\LinesNumbered
\KwIn{$rule$}
\BlankLine
$t=$ find\_or\_insert ($rule.\vec m$);\\
$e=t.table.$insert ($rule.\vec f$);~$e.rule = e.hint = rule$;\\
\If{$k=$ leave\_marker ($e,~t.prev$)}{
    $e.hint = k.hint.pri > rule.pri$ ? $k.hint$ : $rule$
}
report\_hint ($e$);\\
\end{algorithm}

\subsubsection{Rule deletion}
\label{sect:tc:update:delete}

When deleting a rule $r$,  we shall delete it from a corresponding entry $e$ and update the hint of $e$. It starts by looking for $e$ in a tuple. The deletion process will terminate if no entry $e$ is found. 
Otherwise, the  rule will be cleared from $e$. If $e$ is not a marker, it will be deleted directly. If $e$ is a marker, it will update its hint and report the change to its owner. If $e$'s marker exists, $e$'s hint will be set as its marker's hint. Otherwise $e$'s hint will be cleared.

\begin{algorithm}[htbp]
\caption{\label{algo:delete:rule} delete a rule with \emph{TupleChain}}
\small
\LinesNumbered
\KwIn{$rule$}
\BlankLine
\If{($t=$ find ($rule.\vec m$)) \textbf{AND} ($e=t.table.$search ($rule.\vec f$))}{
    \If{$e.rule.$equals ($rule$)}{
        \If{$e.owners$ is empty}{
             delete\_marker ($e,~t.prev$);\\
             delete\_tuple\_if\_empty($t.$erase ($e$));\\
        }
        \Else{
             $k = $ obtain\_marker ($e, ~t.prev$);\\
             $e.hint = k$ ? $k.hint$ : NULL;~$e.rule=$ NULL;\\
             report\_hint ($e$);\\
        }
   }
}
\end{algorithm}

\subsubsection{Marker management}
\label{sect:tc:update:marker}

Marker is used for maintaining tuple chains. Its management is related to rule insertion and deletion. The process of maker creation is more involved. As described in \lag{leave:marker}, the procedure of leaving a marker is recursively performed, until reaching an existing entry, or finding no preceding tuple any more. Finding or deleting a marker is simpler, where just one hash table operation is required.

\begin{algorithm}[htbp]
\caption{\label{algo:leave:marker} leave the marker for an entry}
\small
\LinesNumbered
\KwIn{$e$ /* leave the marker for this entry */}
\KwIn{$t$ /* the target tuple of $e$'s marker */}
\KwOut{$k$ /* return this as $e$'s marker */}
\BlankLine
\If{$t$ \textbf{AND} ($k=t.table.$find ($e.\vec f~\&~t.\vec m$)) is NULL}{
     $k=t.table.$insert ($e.\vec f~\&~t.\vec m$);\\
     $k'$ = leave\_marker ($k, t.prev$);\\
     $k.hint = k'$ ? $k'.hint$ : NULL;\\
}
$k.onwers.$add ($e$);\\
\end{algorithm}

\subsubsection{Hint management}
\label{sect:tc:update:owner}
Once a marker's hint is updated, the change must be reported to its owner(s).
To avoid search through the whole tuple for  the owners of a given marker, as introduced in Section~\ref{sect:tc:graph}, we store all the owners of a marker in a list to trade the memory for the update speed. Fortunately, this is fairly cost-effective (the detailed analysis is presented in section~\ref{sect:analysis:memory}). 


\subsubsection{Tuple insertion / deletion}
\label{sect:tc:update:tuple}
Tuple insertion or deletion will be triggered when its first rule is inserted or all rules have been deleted.
With a chain maintained as a red-black tree, inserting or deleting
a tuple is faster than performing lookup, as no hash computation is required.
For tuple deletion, no additional operation is required other
  than removing the tuple from the chain. However, to insert
a newly created tuple, we may have to probe all existing chains to
determine which one it should be inserted into (some greedy strategies
of selecting chains will be introduced in
section~\ref{sect:construct-adaptation}), or create a new chain for it
if no one can host it. A newly created tuple is empty at its
  insertion. After it is inserted into a chain, every entry of its
  succeeding tuple (if any) will leave their markers in this tuple,
  which may trigger recursive marker insertion in \lag{leave:marker}.
\section{Complexity Analysis}
\label{sect:analysis}
In this section, we make a comprehensive theoretical analysis to
understand how effective \emph{TupleChain} will be and where its
bottlenecks are. These analyses will serve as a guideline for us to
improve its performance.

Suppose $n$ $d$-field rules fall into $m$ tuples, and the tuples form
a tuple graph, which is then broken into $l$ chains. Among these
  chains, the ``biggest'' one (which holds the largest
  number of rules) contains $n'$ rules, and the ``longest'' one is made up
  of $m'$ tuples.
  We analyze the performance of our \emph{TupleChain} scheme accordingly.

Same as most TSS inspired schemes, the unit operation of flow
  lookup and rule updates with \emph{TupleChain} is the hash with
  $d$-field keys. The cost of this operation linearly scales with the
  number of fields, so does the cost of storing $d$-field rules.
  We ignore the parameter $d$ in the following
  analyses for simplification. The results of our evaluations are
compared with other schemes in Table~\ref{tb:complexity}.

\begin{table}[tbp]
\caption{Complexity Comparison}
\begin{center}
\setlength{\tabcolsep}{0.6em}
\renewcommand{\arraystretch}{1.5}
\begin{tabular}{|c|c|c|c|c|}
\hline
  & lookup & memory & \multicolumn{2}{|c|}{update}\\ \cline{4-5}
& & & average & worst \\
\hline
\textit{TSPS}& \lcom{$d\times m$}
& \lcom{$n$} & \multicolumn{2}{|c|}{\lcom{$d$}}\\
\hline
\textit{PTS}& \lcom{$d\times m$}
& \lcom{$d\times n$} & \multicolumn{2}{|c|}{\lcom{$d\times n$}}\\
\hline
\textit{TM}& \lcom{$d\times m''$}
& \lcom{$n$} & \multicolumn{2}{|c|}{\lcom{$d\times m''$}}\\
\hline
\textit{TSBH}& \lcom{$d\times m^{\log_32}$}
& \lcom{$m\times n$} & \multicolumn{2}{|c|}{\lcom{$m\times n$}}\\
\hline
\textit{TC}  & \lcom{$d\times l\times\log_2\frac{m}{l}$}
& \lcom{$m\times n'$} & \lcom{$d\times \frac{m}{l}$} & \lcom{$m'\times n'$}\\
\hline
\multicolumn{5}{l}{$^{\mathrm{a}}$$m'' < m$ but an addtiional linear probe is
  required for searching a tuple.}
\end{tabular}
\label{tb:complexity}
\end{center}
\end{table}

\subsection{Time Complexity of Flow Lookup}
\label{sect:analysis:lookup}
Flow lookup with \emph{TupleChain} needs to search
all $l$ chains of $m$ tuples, with a binary
search on each chain. We denote the lookup cost as $F(m,
l)$, which is the number of tuples that will be visited with a
  hash probe.

\begin{theorem}\label{the:tc:lookup}
$F(m, l)$ has an upper bound $(l\times (1 + \log_2(\frac{m}{l})))$.
\end{theorem}

\begin{proof}
Suppose the $i$-th chain has $m_{i}$ tuples, and its lookup cost
  is denoted as $F_{i}(m,l)$. Because of binary search,
$F_{i}(m,l) \le (1 + \lfloor \log_2m_{i} \rfloor)$. For $l$
  chains in total, we have
\begin{align}\label{eq:tc:lookup:1}
    F(m,l) = \sum_{i=1}^{l}F_{i}(m,l) & \le \sum_{i=1}^{l}(1 + \lfloor \log_2m_{i} \rfloor) \notag\\
    &\le l + \sum_{i=1}^{l}\log_2m_{i}\notag\\
    & = l +  \log_2\prod_{i=1}^{l} m_{i}
\end{align}
According to the arithmetic-geometric average inequality,
\begin{align}\label{eq:tc:lookup:2}
  \prod_{i=1}^{l} m_{i} \le (\frac{\sum_{i=1}^{l}m_{i}}{l})^{l} = (\frac{m}{l})^{l}
\end{align}
Combining \lieq{tc:lookup:1} and \lieq{tc:lookup:2}, we get $F(m,l) \le l\times (1 + \log_2(\frac{m}{l}))$.
\end{proof}

To gain more insight into this upper bound, we treat $m$ as a
constant to analyze the function $f(l)=l\times (1 +
\log_2(\frac{m}{l}))$, where $l \in [1, m]$. Its first-order and second-order
derivatives are
  \begin{align*}\label{eq:tc:lookup:trend}
    f'(l) &=-\log_2l + \log_2(\frac{2m}{e})\\
    f''(l) & = -\log_2e \times l^{-1}
  \end{align*}
$f^{''}(l) < 0$ always holds, so $f'(l)$ is monotonically decreasing.
Letting $f'(l)=0$, we get $l=\frac{2m}{e}$. Accordingly, $f'(l)>0$ holds
  when $l$ increases from 1 to $\frac{2m}{e}$, so $f(l)$
monotonically increases from $f(1)$ to $f(\frac{2m}{e})$. While $l$
continuously increases to $m$, $f'(l)<0$ holds instead, and $f(l)$
monotonically decreases to $f(m)$. A sketch of
$f(l)$ is shown in Fig.~\ref{fig:fnc}.

\begin{figure}[tbp]
\centerline{\includegraphics[width=.5\linewidth]{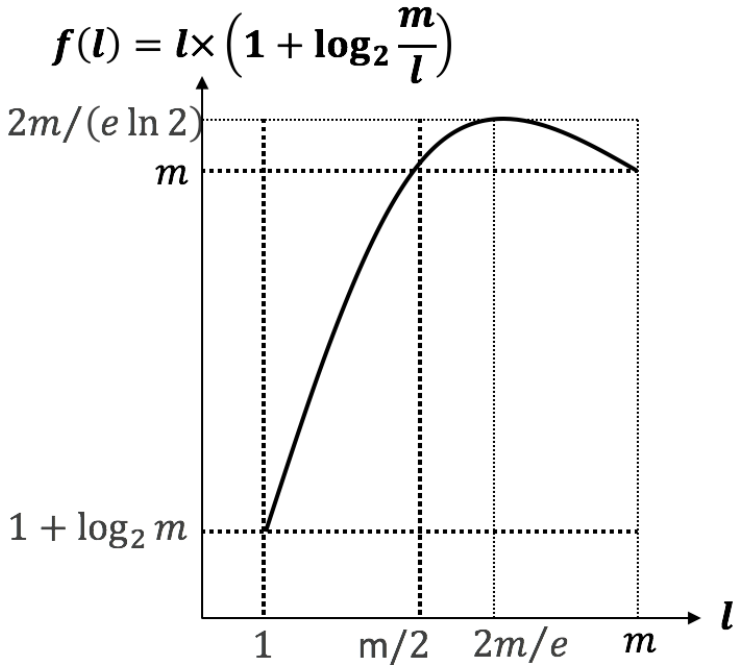}}
\caption{A sketch of the function $f(l)=l\times (1 + \log_2(\frac{m}{l}))$.}
\label{fig:fnc}
\end{figure}

Noted from the curve, $f(l) = m$ has two roots across $[1, n]$. It is
easy to verify that they are $l=\frac{m}{2}$ and $l=m$ respectively.
Since $\frac{m}{2} < \frac{2m}{e}$, $f'(l)>0$ holds across $l\in[1,
\frac{m}{2}]$, namely $f(l)$ monotonically increases. Then, we have:

\begin{corollary}\label{cor:tc:lookup:complexity}
When $l < \frac{m}{2}$ holds, \emph{TupleChain}'s lookup complexity is
$\mathcal{O}(l \times \log_2\frac{m}{l})$.
\end{corollary}
\begin{proof}
$l < \frac{m}{2} \Longrightarrow \log_2\frac{m}{l} > 1$, then, $F(m,l)
< 2l \times \log_2\frac{m}{l}$, thus $\mathcal{O}(F(m,l))=\mathcal{O}(l \times \log_2\frac{m}{l})$.
\end{proof}

\begin{corollary}\label{cor:tc:lookup:best} Once the tuple graph can be broken into chains with each having fewer than half
  the number of tuples (i.e., $l < \frac{m}{2}$), \emph{TupleChain}
  offers a lower lookup complexity than \emph{TSS}
  ($\mathcal{O}(m)$), and the fewer the number of chains, the lower
the lookup complexity.

\end{corollary}


\subsection{Total Space Complexity}
\label{sect:analysis:memory}
To evaluate the space complexity of \emph{TupleChain}, we begin
  with the analysis on a single chain with $n_{c}$ rules falling into
  $m_{c}$ tuples. Its space complexity can be evaluated via the summation of
the number of entries inside all tuples on the chain for two reasons. First, the entries of a tuple are stored and managed by a hash table, and the storage for all tuples' hash tables on a chain dominates the space cost. Second, the space consumed by each entry is also related to the number of entries. Every entry of a tuple is designed to have the
same length, and is associated with a linked list that stores pointers to its owners.
Since one entry owns one marker at most, any entry
could be counted as an owner at most once. Accordingly, the total length
of all owner lists at most equals to the total number of
entries. Therefore, we only need to focus on the number of entries created.

First of all, every rule takes up an entry, and the process
of leaving markers will create additional entries. Leaving the marker
for a rule is recursively performed tuple by tuple, which may produce
at most $m_{c}-1$ entries. Therefore, there may be at most $n_{c} +
n_{c} \times (m_{c}-1) = n_{c} \times m_{c}$ entries. So the
  space complexity is \lcom{$n_{c} \times m_{c}$}. In \emph{TupleChain},
  every chain is independent, so for all chains in total, we have:

\begin{theorem}\label{the:tc:memory}
TupleChain's space complexity is \lcom{$n' \times m$}.
\end{theorem}

\subsection{Time Complexity of Rule Update}
\label{sect:analysis:update}
Tuple insertion/deletion are actually
rarely triggered\footnote{the rates of tuple insertion/deletion we observed
  throughout our experiments were as low as $0.1\%$.} and can be
performed efficiently. Therefore, we do not count them for
complexity analysis. We focus on the complexity of handling
markers and hints when inserting / deleting a rule within an
existing tuple, as the corresponding operations are the
most time-consuming.


Inserting / deleting a rule with a tuple only affects a
  single chain that hosts this tuple. We denote the number of tuples
  on this chain and the number of rules belonging to the tuples on the chain as
  $m_{c}$ and $n_{c}$ respectively. Since any entry has one marker at most and leaving
   the marker for a entry is recursively performed tuple by tuple, at most
  $m_{c} - 1$ entries will be accessed or created. As obtaining or deleting
  a marker only requires one hash operation, the time
  complexity of marker maintenance turns to be \lcom{$m_{c}$}.

 Now we evaluate the time complexity for reporting hints. By
  associating every entry with a separate owner list, we can locate all
  owners of an entry quickly without any hash operation. In addition, the
  hint reporting starting from an entry in tuple $t_{i}$ is also performed
  tuple by tuple recursively, forming a \emph{reporting tree} with
  every level of entries residing in a tuple $t_{j} (j > i)$
  excluding the tree root which is in $t_{i}$. It is a
  tree rather than a path because one entry can have multiple
  owners. For an entry in $t_{i}$, its reporting tree excluding
    this entry can be as large as covering all entries in all tuples
  $t_{j} (j >  i)$. This determines that the worst case time complexity of hint reporting is
  \lcom{$n_{c} \times m_{c}$}.

In the average case, however, the cost of reporting hints
is perfectly amortized. Suppose the tuple $t_{i}$ ($i\in[1, m_{c}]$)
  has $x_{i}$ entries, we evaluate the cost of reporting hints for all
  entries in this tuple. Since one entry has one marker at most, any
  two reporting trees rooted at two different entries in $t_{i}$ would never
  intersect. Accordingly, the union of all reporting trees rooted at
  $t_{i}$ will have $\sum_{j=i+1}^{m_{c}}x_{j}$ entries at most, which
determines the cost of reporting hints for all entries in this
tuple. Thus, the total cost across $m_{c}$ tuples is summed up as:
\begin{align}
\sum_{i=1}^{m_{c}-1}\sum_{j=i+1}^{m_{c}}x_{j}
= \sum_{i=1}^{m_{c}}(i-1)\times x_{i}
< m_{c} \times \sum_{i=1}^{m_{c}}x_{i} \notag
\end{align}
Since this cost can be amortized by all $\sum_{i=1}^{m_{c}}x_{i}$
entries on this chain, the average-case time complexity for hint
reporting turns to be \lcom{$m_{c}$}.

Any update will be performed within one of the chains in
  \emph{TupleChain}, so we take $n'$ and $m'$ to calculate the overall
update complexity.
\begin{theorem}\label{the:tc:update}
TupleChain's update complexity is \lcom{$m'$} in the
average case, and \lcom{$n'\times m'$} in the worst case.
\end{theorem}
\section{Boosting Practical Performance}
\label{sect:optimization}
The comprehensive theoretical analysis in the last section provides us
with more insight into \emph{TupleChain}, which enables us to refine
the design to boost its practical performance.

\subsection{Optimal Chain Construction}
\label{sect:construct-adaptation}
According to \lcor{tc:lookup:best}, to construct a
  \emph{TupleChain} with the lowest lookup complexity, we should break
  the tuple graph into a minimal number of chains. This is
essentially a classic problem in graph theory known as \emph{minimum
  path cover}, which is NP-hard~\cite{mpc}. However, for a directed
acyclic graph (DAG) like the tuple graph, it can be solved as a
matching problem. We adopt the Hungarian
algorithm~\cite{hungarian} to solve this problem, and construct an
optimal \emph{TupleChain} accordingly.

\subsection{Greedy Strategies for Tuple Insertion}
\label{sect:tuple-insertion}
Once a new tuple is created, we first try to insert it into an
existing chain whenever feasible to control the number of chains,
which is the key factor to restrict the lookup cost
  (\lcor{tc:lookup:best}). In case that multiple chains can host
this tuple, we chose the shortest one to control the length
of the longest chain, as it determines the memory cost and
update overhead (\lthe{tc:memory}). Further, if there are multiple chains that
  can host this tuple, we choose the one with fewer rules to reduce
the worst-case update overhead (\lthe{tc:update}).
\subsection{An Extension by Rule Grouping}
\label{sect:etc}
Many researchers have observed that the number of tuples grows significantly when the number of rules become larger \cite{TM, TupleTree, MT}. This will cause too many and too long chains in our scheme, resulting in the decrease of performance. In view of this, we propose a new data structure, {\emph{Extended TupleChain} (ETC in short), to boost the overall performance during practical running.

\begin{figure}[tbp]
  \centering
    \includegraphics[width=\linewidth]{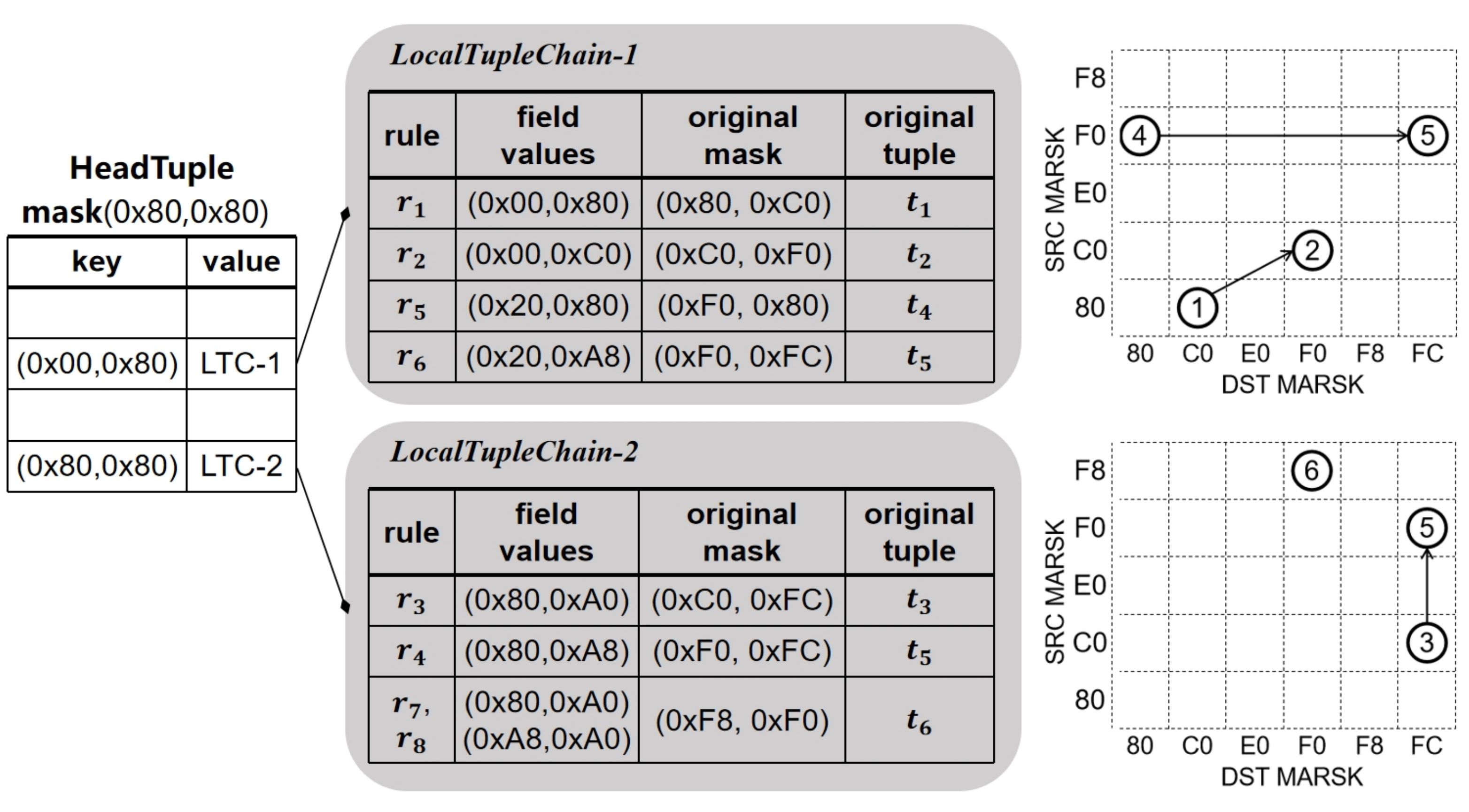}
    \caption{An Extened Tuple Chain with one head tuple.}
    \label{fig:etc}
\end{figure}

Our key idea is to reduce the access of tuples by merging chains into groups. Nearby chains in the tuple graph will be merged into one group. 
For each group, we set up a mask by taking the intersection of all masks from the chains to merge. We create a head tuple with this mask and insert all rules of these chains  into this head tuple. With keys formed by applying the mask of the head tuple to the fields of rules,  several rules may fall into the same entry. In Fig.~\ref{fig:etc},  masked by $\vec m_h$ (0x80,0x80) of head tuple to form the key  (0x00, 0x80), $r_1$ with (0x00, 0x80), $r_2$ with (0x00, 0xC0),  $r_5$ with (0x20, 0x80) and $r_6$ with (0x20, 0xA8) are inserted into the same entry.  These rules will be further constructed into a local TupleChain, similar to that of the \emph{TupleChain.}

Compared to the original TupleChain, local TupleChains have much shorter lengths.
This can boost the performance in all aspects. First of all, any rule update is processed within one ``small'' instance, so the update overhead can be sharply reduced. Although all head tuples must be probed for a lookup, only one ``small'' instance managed by every head tuple that yields a match will be checked intensively. Therefore, the lookup cost can also be significantly reduced, especially when there are only a few head tuples (which is the fact in practice according to our evaluations). At last, the total memory footprint of ETC can be reduced as well. The maintenance of additional markers that contains no rule dominates the storage of TupleChain, which can be greatly reduced with shorter chains in ``smaller'' \emph{TupleChain}.

\section{Performance Evaluations}
\label{sect:exp}
In this section, we evaluate the performance and scalability
  of TupleChain (\emph{TC}) and its extension \emph{ETC} with extensive
  experiments, and compare them with the state-of-the-art flow lookup
  schemes as well as classical packet classification algorithms. We
  implement all algorithms on our own, except \emph{TM}~\cite{TM}, \emph{MT}\cite{MT}
  \emph{TupleTree}~\cite{TupleTree}, and
    \emph{CutTSS}~\cite{CutTSS}, whose codes are downloaded from public
    repositories of
    GitHub\footnote{https://github.com/drjdaly/tuplemerge;
          https://gitee.com/dave\_ta/TupleTree;\newline https://github.com/zcy-ict/MultilayerTuple;http://www.wenjunli.com/CutTSS;} and provided by the authors
    respectively. Our evaluation platform consists of a \emph{lookup module} running with
  a thread for flow lookup (or packet classification) and rule
  updates, a \emph{tester} and an \emph{update manager} that run in
  different threads to feed the lookup module with packets and update
  requests respectively, via shared buffers at pre-defined yet configurable rates.


\label{sect:exp:tuple}
\begin{table}[tbp]
\caption{Number of tuples accessed in average}
\begin{center}
\setlength{\tabcolsep}{0.4em}
\renewcommand{\arraystretch}{1}
\begin{tabular}{|c|c|c|c|c|}
\hline
  & \multicolumn{2}{|c|}{1 kilo rules} & \multicolumn{2}{|c|}{1 million rules}\\ \cline{2-5}
  & number of tuples & other$^{\mathrm{a}}$ & number of tuples & other$^{\mathrm{a}}$ \\
\hline
\textit{TSS}               & 80      &        & 404     &        \\ \hline
\textit{PTS}               & 1        & 5.8  & 1.9      & 5.3   \\ \hline
\textit{TSBH}             & 28.2   &        & 83.4    &         \\ \hline
\textit{PSTS}             & 55.8    &        & 334     &         \\ \hline
\textit{TM}                & 1        & 1     & 8         & 45.6  \\ \hline
\textit{MT}               & 1         &       1&  5.6&           8.3\\ \hline
\textit{TupleTree}         & 1         &       1&  5.2&           6.5\\ \hline
\textit{TC}                & 12.1    &        &  23.1   &          \\ \hline
\textit{ETC}               & 1         &       &  4.2     &           \\ \hline
\multicolumn{5}{l}{$^{\mathrm{a}}$for \emph{PTS}, it's the number of accessed trie nodes;}\\
\multicolumn{5}{l}{~for \emph{TM}, \emph{MT}, \emph{TupleTree}, it's the number of verifications;}\\
\end{tabular}
\label{tb:tuples}
\end{center}
\end{table}

\subsection{Reduction of the Number of Tuples to Search}
All \emph{TSS}-based schemes attempt to reduce the
number of tuples to search. We compare 9 schemes, \emph{TSS},
  \emph{PTS}, \emph{TSBH}, \emph{PSTS}, \emph{TM}, \emph{MT}, \emph{TupleTree}, \emph{TC}, and
  \emph{ETC} using two datasets with 1 kilo and 1 million 2-field rules
  and corresponding traffic traces respectively. The average number of
  tuples accessed for one lookup is reported in
  Table~\ref{tb:tuples}. \emph{TSS} produced $80$ and $404$ tuples
respectively, which are all searched in a lookup. Our results
confirm the statement claimed in~\cite{tss} that \emph{PTS} has a
  promising practical performance. In this study, it requires fewer
than $2$ tuple searches on average. However, before tuple search, it
needs to process each field with prefix trees, and combine the
results via bitmap operations. On the other hand, \emph{TM} reduces
the number of tuples of two datasets from $80$ to $1$ and from $404$
to $8$ at the cost of additional verifications. \emph{MT} and \emph{TupleTree} also reduce the number of tuples
and their additional verifications are smalller than \emph{TM} in 1 million rules case.

%
Compared with \emph{TSS} on these two datasets, our basic
\emph{TupleChain} scheme reduces the number of tuples to search by
$\%84.9$ and $\%94.3$ respectively, which can be further improved by
its extension \emph{ETC}. \emph{ETC} requires fewer
tuple searches than all other approaches except
  \emph{PTS}. Although \emph{PTS} requires slightly fewer tuple searches than
  \emph{ETC}, it brings in additional cost on tree traversals. It
    is clear that \emph{ETC} is a better choice for practical
    implementation compared with \emph{TC}, but \emph{TC} guarantees
    the worst-case performance of \emph{ETC}.

\subsection{Performance with Regular Rules via ClassBench}
\label{sect:exp:classbench}
We compare the performance of \emph{ETC} with
three state-of-the-art schemes for fast packet classification,
  \emph{MT}, \emph{TupleTree} and \emph{CutTSS}, because of their outstanding performance. We conduct performance
  evaluations using the rules and traffics generated by
  ClassBench~\cite{ClassBench}. There are 1000 fw (firewall) rule sets and
   1000 acl (access control list) rule sets with different configuration, each consisting of
    100 kilo rules and a corresponding traffic trace, are used
  for evaluation. We measure the lookup speed of each approach in \emph{million
    packets per second (MPPS)}  and draw the
  complementary cumulative distribution function 
  accordingly. Figures~\ref{exp:ccdf:fw} and \ref{exp:ccdf:acl} show
  the results with fw and acl rules respectively. It's clear that
  \emph{ETC} outperforms other two in most cases, more significant  with fw rules. As for
acl rules, the four curves are close. The performance of \emph{MT} and \emph{TupleTree} is improved, because the distribution of acl rules is concentrated and beneficial for merging algorithms.

\begin{figure}[tbp]
  \centering
  \begin{minipage}[b]{0.48\linewidth} 
    \centering
    \includegraphics[width=\linewidth]{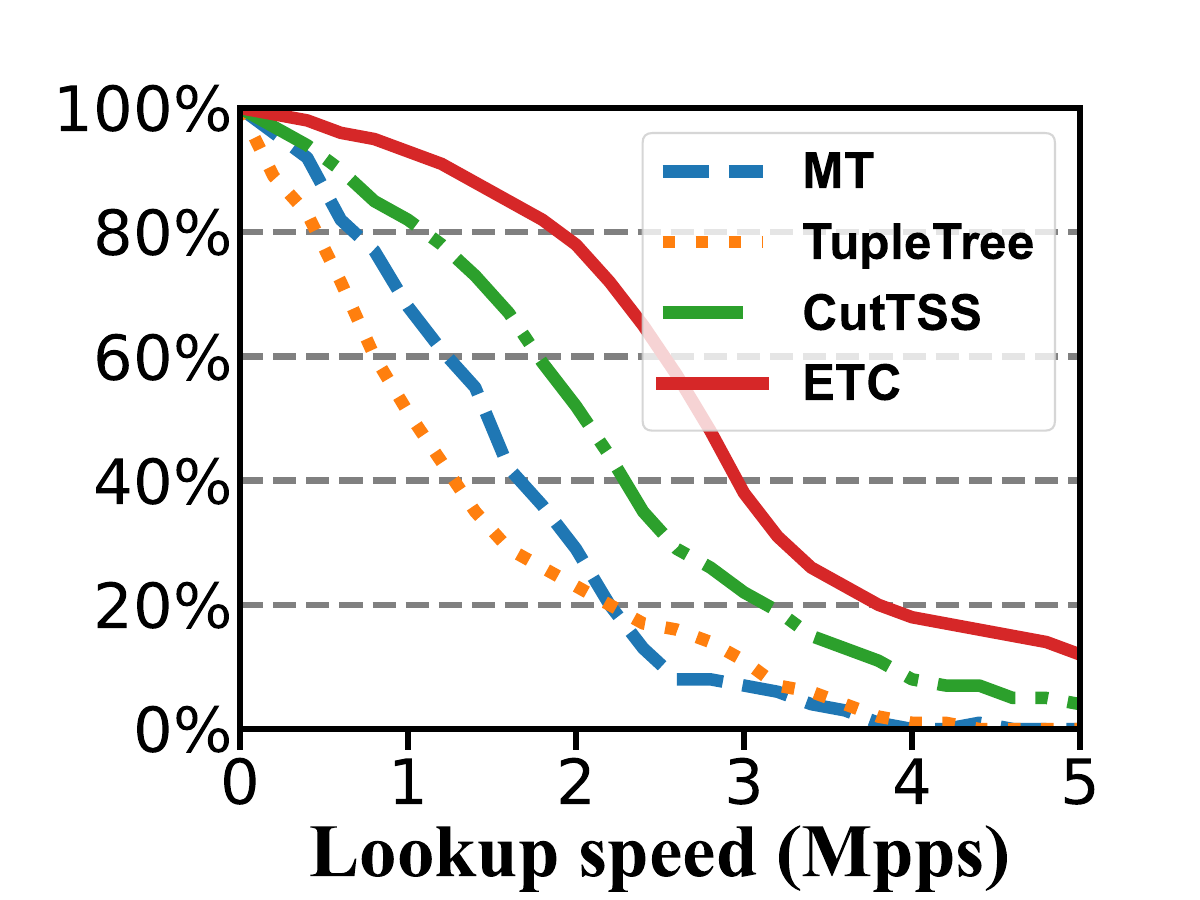}
    \caption{CCDF of lookup speed with 10k fw rules and traces.}
    \label{exp:ccdf:fw}
  \end{minipage}%
  \quad
  \begin{minipage}[b]{0.48\linewidth}
    \centering
    \includegraphics[width=\linewidth]{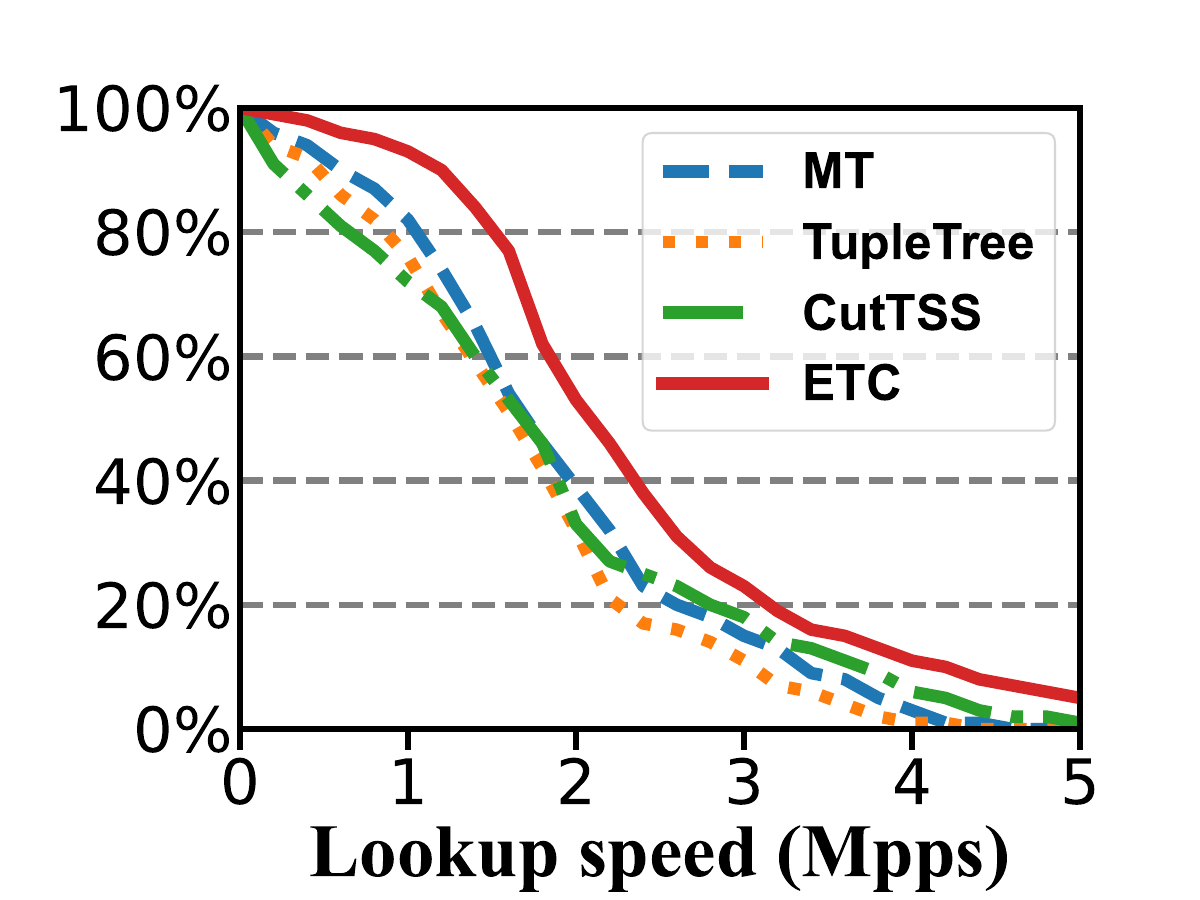}
    \caption{CCDF of lookup speed with 10k acl rules and traces.}
    \label{exp:ccdf:acl}
  \end{minipage}
\end{figure}

\subsection{Scalability of Algorithms}
\label{sect:exp:scalbility}
Finally, we demonstrate the scalability of \emph{ETC} in four
  challenging scenarios, and compare it with \emph{MT}, \emph{TupleTree} and 
  \emph{CutTSS} that have claimed their scalability in at least one of
these scenarios.

\subsubsection{Scalability to the Number of Fields}
\label{sect:exp:field}
We evaluate 4 schemes with 50 datasets sized around
$100~K$, where the number of fields ranges from $2$ to $100$. With the
tester flushing packets at $10$ MPPS, we measure the system
throughput and memory cost for each scheme. As shown in Fig.~\ref{fig:exp:field:throughput} and
Fig.~\ref{fig:exp:field:memory}, only \emph{ETC} works in all
cases. The others experience a sharp decline
in throughput. The throughput of each drops below $0.01$
MPPS once the number of fields exceeds $20$. In contrast, \emph{ETC} shows excellent
scalability, with its throughput only decreasing from
$6$ MPPS to $1.1$ MPPS. For the memory cost, others require
more than $1~GB$ of memory, and can not be constructed when there are
more than $50$ fields (the system runs out of memory). In contrast, \emph{ETC}
requires less than $70~MB$ of memory to handle $100~K$ $100$-field
rules.

\begin{figure}[tbp]
  \centering
  \begin{minipage}[b]{0.5\linewidth} 
    \centering
    \includegraphics[width=\linewidth]{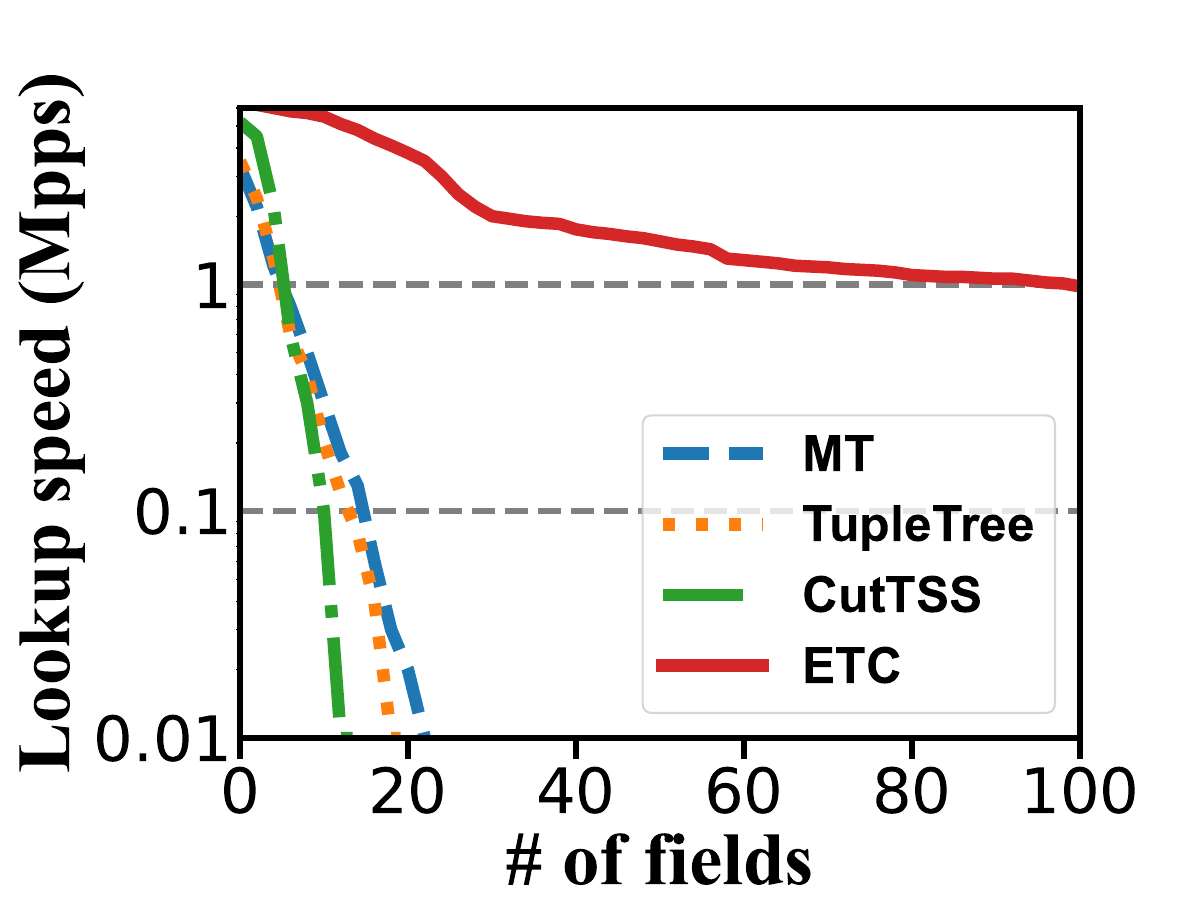}
    \caption{performance versus \# of fields}
    \label{fig:exp:field:throughput}
  \end{minipage}%
  \begin{minipage}[b]{0.5\linewidth}
    \centering
    \includegraphics[width=\linewidth]{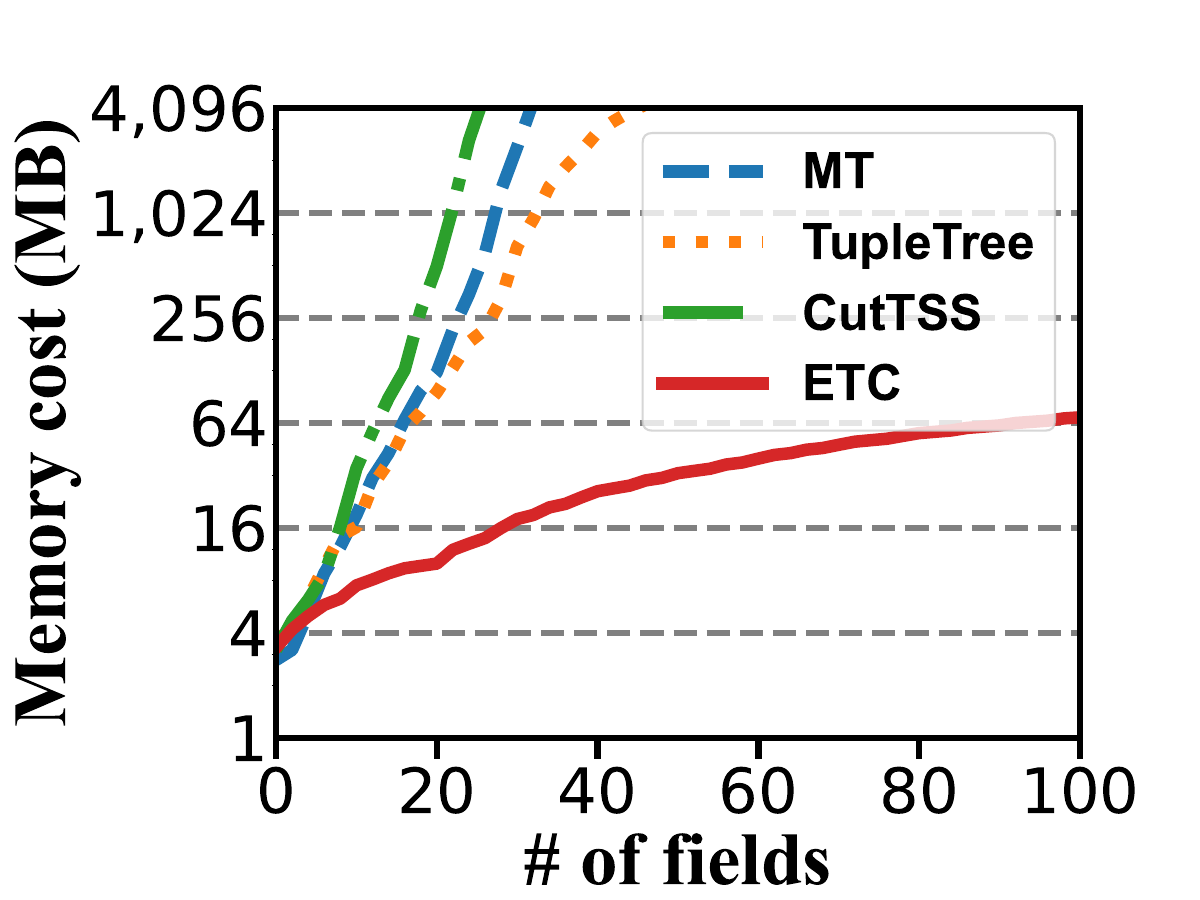}
    \caption{memory versus \# of fields}
    \label{fig:exp:field:memory}
  \end{minipage}
\end{figure}

\subsubsection{Scalability to the Size of Dataset}
\label{sect:exp:rule}
We evaluate each of the 4 schemes with 6 2-field rule sets of different
sizes to measure its system throughput and memory cost.
\emph{ETC} achieves the highest performance in all cases (as shown in
Fig.~\ref{fig:exp:scale:throughput}), while \emph{MT} has the lowest memory cost
(Fig.~\ref{fig:exp:scale:memory}). Overall, in comparison to
\emph{MT}, \emph{ETC} achieves a speedup of $1.25 \sim 3.7$ at the cost
of only $\%10 \sim \%30$ additional memory consumption. Additionally,
\emph{ETC}'s throughput decreases the slowest as the scale increases. Only
\emph{ETC} can offer a throughput higher than $1$ MPPS to process
a data set with $10$ million rules. We can see that the performance decreases and the memory cost increases sharply for \emph{CutTSS}, which is due to the copy of rules in the decision tree algorithm.

\begin{figure}[tbp]
  \centering
  \begin{minipage}[b]{0.5\linewidth} 
    \centering
    \includegraphics[width=\linewidth]{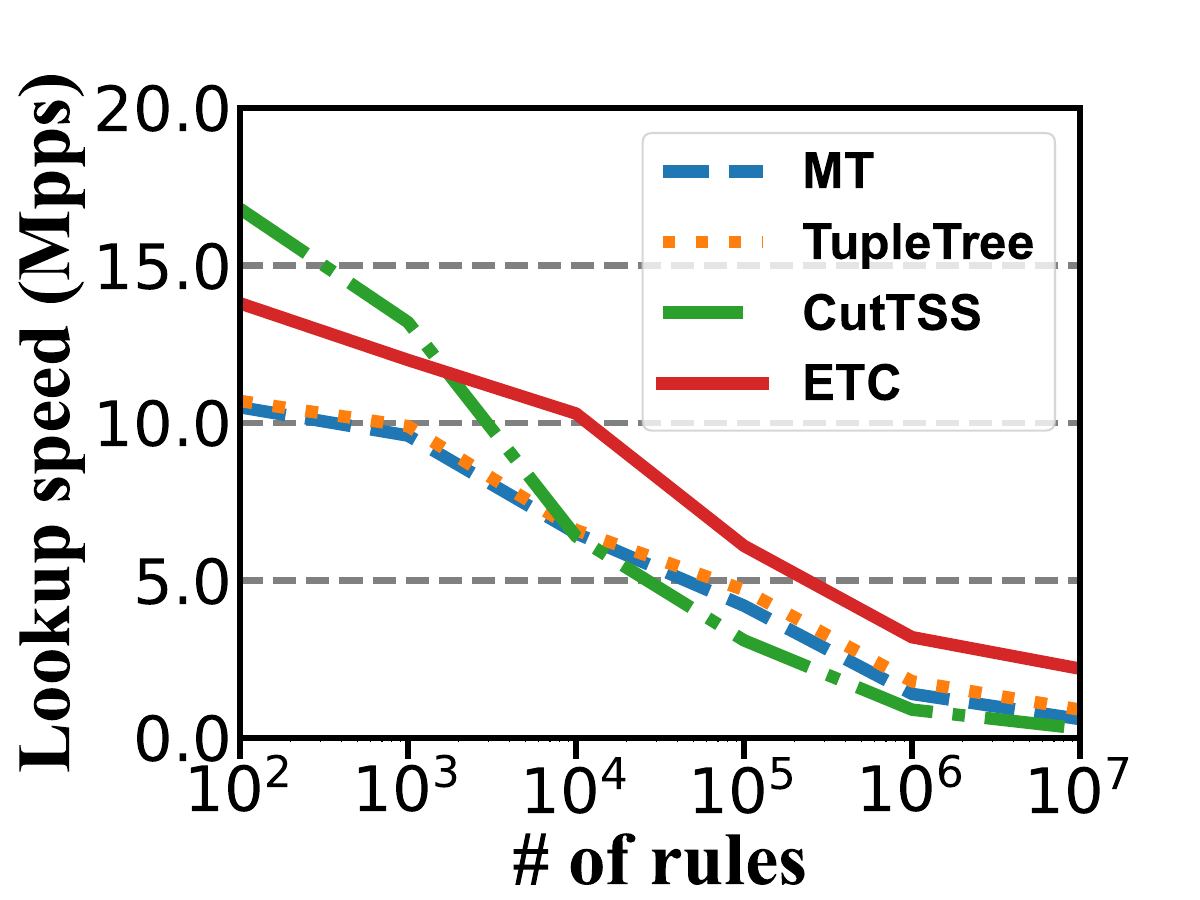}
    \caption{performance versus scale}
    \label{fig:exp:scale:throughput}
  \end{minipage}%
  \begin{minipage}[b]{0.5\linewidth}
    \centering
    \includegraphics[width=\linewidth]{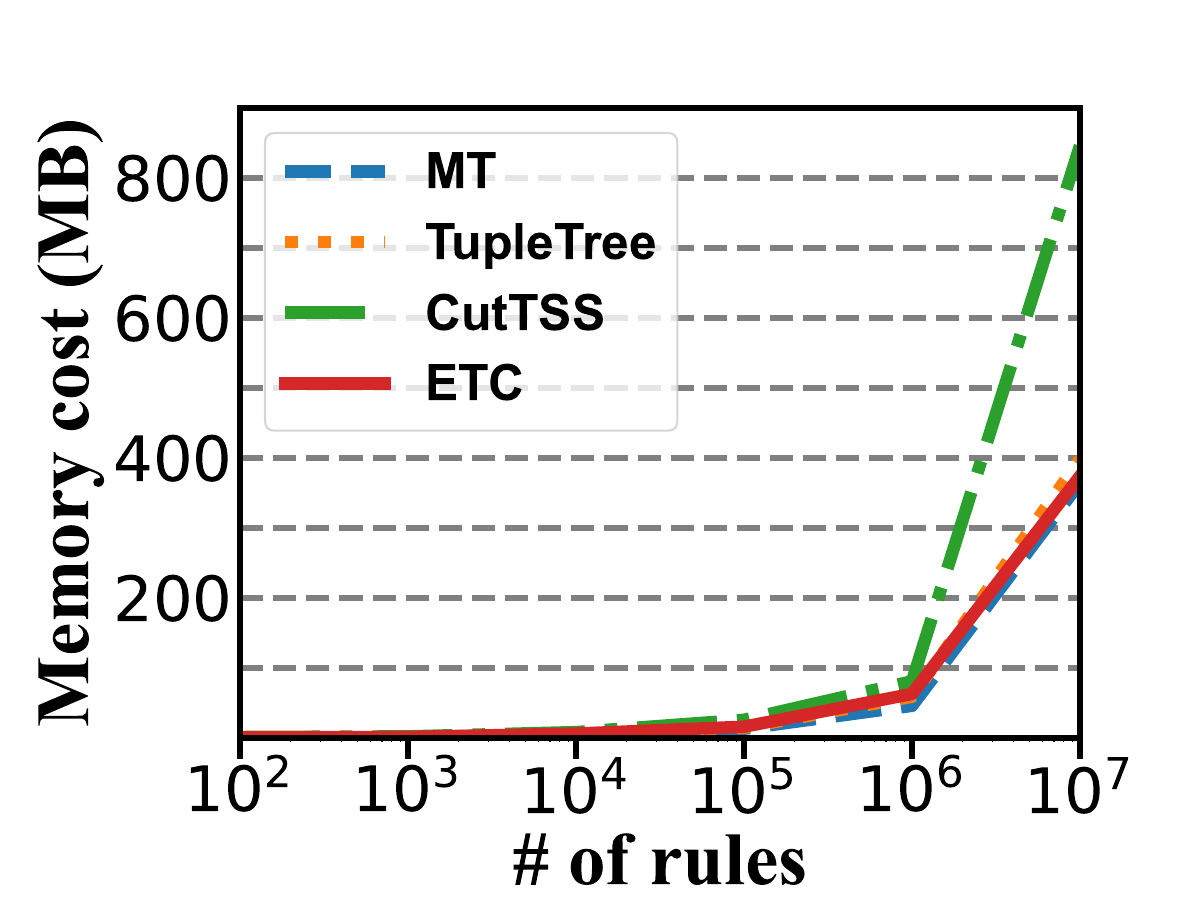}
    \caption{memory cost versus scale}
    \label{fig:exp:scale:memory}
  \end{minipage}
\end{figure}

\subsubsection{Scalability to the Rate of Receiving Packets}
\label{sect:exp:throughput}
We compare the system throughput of 4 schemes with a
  dataset of 10 million 2-filed rules and with the tester flushing packets at increasing
  rates. Figure~\ref{fig:exp:throughput} shows the similar trend for
  each of them. As the transmission rate of the tester increases, its
receiving rate increases linearly at the beginning and then reaches the peak
around a particular rate. \emph{ETC} can accommodate a maximum
throughput of around $2$ MPPS, with a speed up of $1.2 \sim 2$ compared to
\emph{MT}, \emph{TupleTree} and \emph{CutTSS}.

\subsubsection{Scalability to the Rate of Updates}
\label{sect:exp:update}
We evaluate the 4 schemes with a dataset of 10 million
  2-field rules as well as 1 million insertion/deletion requests. The
  tester flushes packets at a fixed rate of 2 MPPS. As the update rates
increases from $100$ to $10$ million per second, we measure
the receiving rate at the tester. In Fig.~\ref{fig:exp:update}, all
schemes experience a decrease around $20\%$ in
system throughput, but \emph{ETC} remains the fastest all the
  time, staying as fast as $1.4$ MPPS.

\begin{figure}[tbp]
  \centering
    \begin{minipage}[b]{0.5\linewidth}
    \centering
    \includegraphics[width=\linewidth]{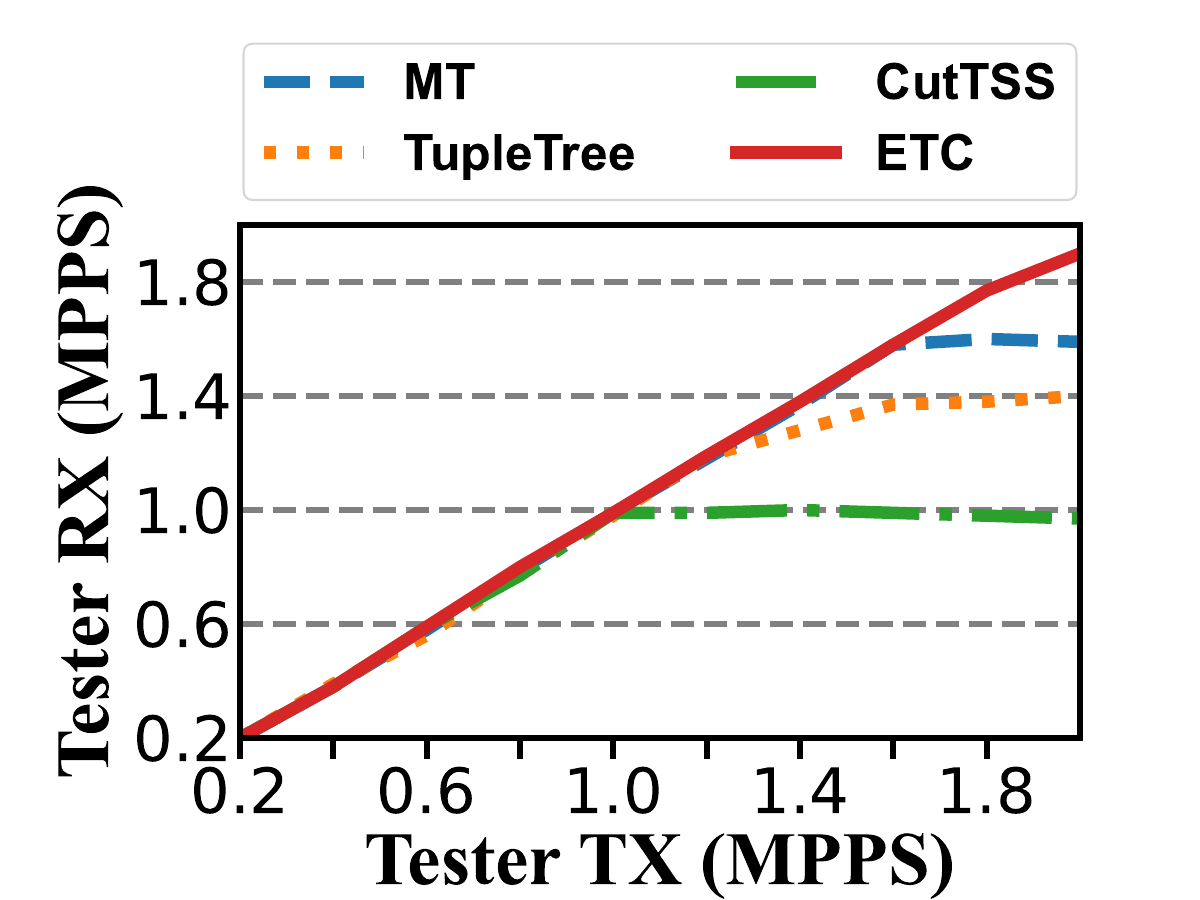}
    \caption{lookup performance with increasing link rates.}
    \label{fig:exp:throughput}
  \end{minipage}%
  \begin{minipage}[b]{0.5\linewidth} 
    \centering
    \includegraphics[width=\linewidth]{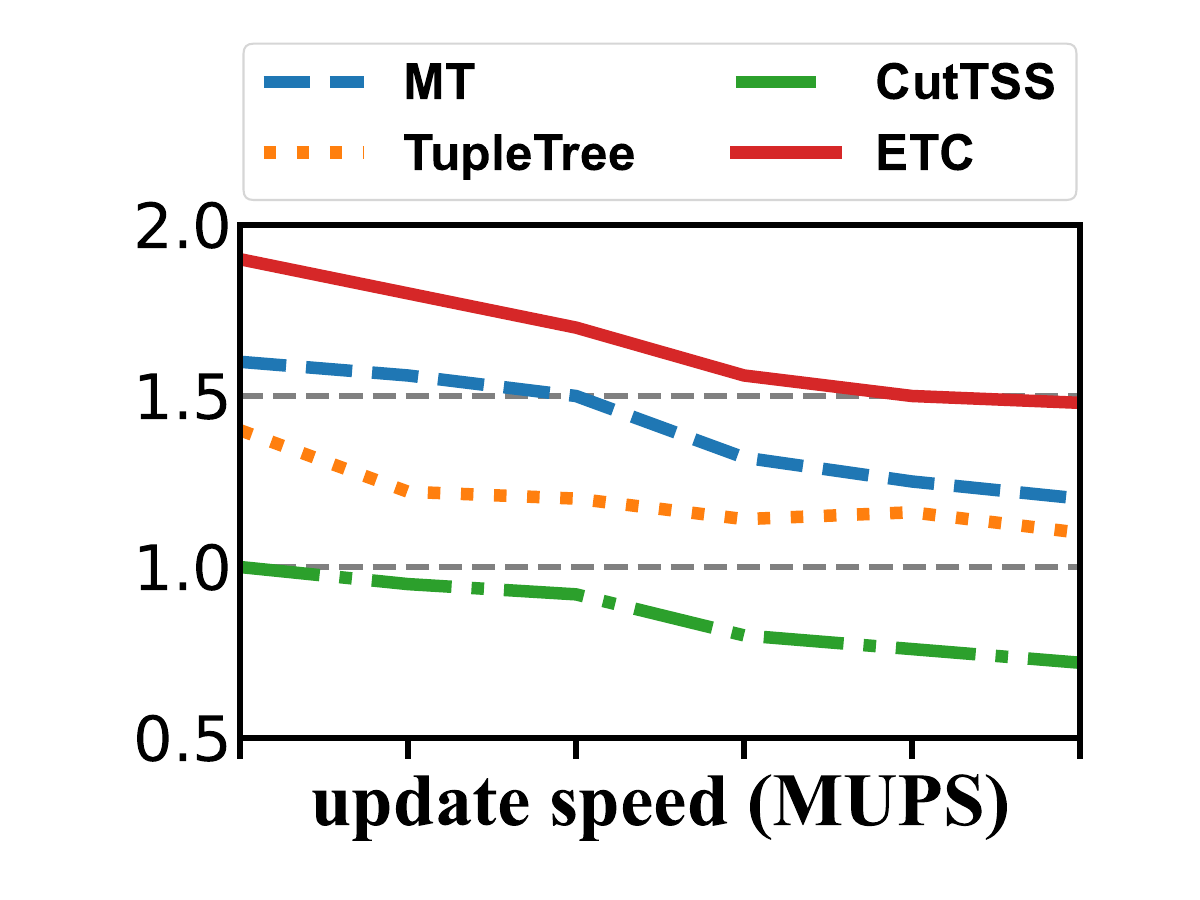}
    \caption{lookup performance with frequent updates.}
    \label{fig:exp:update}
  \end{minipage}%
\end{figure}
\section{Conclusions}
\label{sect:end}

In this paper, we propose a novel scheme for OpenFlow table lookup, with both fast lookup and efficient updates, as well as multifaceted scalability. The key idea under this approach is to explore connections among rule groups (i.e., tuples) to guide more efficient lookup. It is proved to have a near-logarithmic worst-case computing complexity for flow lookup, and a desirable average-case computing complexity for rule updates. Its promising actual performance and scalability are clearly demonstrated via extensive experiments. This work confirms that TSS model is a good starting point for building up a scalable flow lookup scheme. Besides, our experience suggests that a desirable computing complexity might be helpful to achieve good scalability, and that the actual performance can be improved greatly by making better use of the characteristics of rule distribution.

\bibliographystyle{plain}
\bibliography{reference}

\end{document}